\documentclass[11pt]{llncs}
\usepackage[english]{babel} 
\usepackage[T2A]{fontenc}
\usepackage{amsmath}
\usepackage{amssymb}
\usepackage{epic}

\usepackage[usenames]{color}
\usepackage{gastex}

\DeclareSymbolFont{rsfscript}{OMS}{rsfs}{m}{n}
\DeclareSymbolFontAlphabet{\mathrsfs}{rsfscript}

\newcommand{\sw}{synchronizing word}
\newcommand{\sws}{synchronizing words}
\newcommand{\sa}{synchronizing automata}
\newcommand{\san}{synchronizing automaton}

\begin{document}

\title{Approximating the minimum length\\ of synchronizing words is hard}

\titlerunning{Approximating the length of synchronizing words}

\author{Mikhail V. Berlinkov}

\authorrunning{M. V. Berlinkov}

\tocauthor{M. V. Berlinkov (Ekaterinburg, Russia)}

\institute{Department of Algebra and Discrete Mathematics\\
Ural State University\\
620083 Ekaterinburg, Russia\\
\email{berlm@mail.ru}}

\maketitle

\begin{abstract}
We prove that, unless $\mathrm{P}=\mathrm{NP}$, no polynomial
algorithm can approximate the minimum length of \sws\ for a given
\san\ within a constant factor.
\end{abstract}

\section*{Background and overview}

Let $\mathrsfs{A}=\langle Q,\Sigma,\delta\rangle$ be a complete
\emph{deterministic finite automaton} (DFA), where $Q$ is the state
set, $\Sigma$ is the input alphabet, and $\delta:Q\times\Sigma\to Q$
is the transition function. The function $\delta$ extends in a
unique way to a function $Q\times\Sigma^*\to Q$, where $\Sigma^*$
stands for the free monoid over $\Sigma$; the latter function is
still denoted by $\delta$. Thus, each word in $\Sigma^*$ acts on the
set $Q$ via $\delta$. The DFA $\mathrsfs{A}$ is called
\emph{synchronizing} if there exists a word $w\in\Sigma^*$ whose
action resets $\mathrsfs{A}$, that is to leave the automaton in one
particular state no matter which state in $Q$ it starts at:
$\delta(q,w)=\delta(q',w)$ for all $q,q'\in Q$. Any such word $w$ is
called a \emph{\sw} for $\mathrsfs{A}$. The minimum length of \sws\
for $\mathrsfs{A}$ is denoted by $\min_{synch}(\mathrsfs{A})$.

Synchronizing automata serve as transparent and natural models of
error-resistant systems in many applications (coding theory,
robotics, testing of reactive systems) and also reveal interesting
connections with symbolic dynamics and other parts of mathematics.
For a brief introduction to the theory of \sa\ we refer the reader
to the recent survey~\cite{Vo08}. Here we discuss only some
complexity-theoretical issues of the theory. In the following we
assume the reader's acquaintance with some basics of computational
complexity that may be found, e.g., in~\cite{GJ79,Pa94}.

There is a polynomial algorithm (basically due to
\v{C}ern\'y~\cite{Ce64}) that decides whether or not a given DFA is
synchronizing. In contrast, determining the minimum length of \sws\
for a given \san\ is known to be computationally hard. More
precisely, deciding, given a \san\ $\mathrsfs{A}$ and a positive
integer $\ell$, whether or not $\min_{synch}(\mathrsfs{A})\le\ell$
is NP-complete~\cite{Ep90,GK95,Sa03,Sa07}. Moreover, deciding, given
the same instance, whether or not $\min_{synch}(\mathrsfs{A})=\ell$
is both NP-hard and co-NP-hard~\cite{Sa07}. Thus, unless
$\mathrm{NP}=\mathrm{co}$-NP, even non-deterministic algorithms
cannot find the minimum length of \sws\ for a given \san\ in
polynomial time.

There are some polynomial algorithms that, given a \san, find \sws\
for it, see~\cite{Ep90,Tr06,Ro09}. Such algorithms can be considered
as approximation algorithms for calculating the minimum length of
\sws\ but it seems that they have not been systematically studied
from the approximation viewpoint. Experiments show that Eppstein's
greedy algorithm~\cite{Ep90} behaves rather well on average and
approximates $\min_{synch}(\mathrsfs{A})$ within a logarithmic
factor on all tested instances; however, no theoretical
justification for these observations has been found so far.

In this paper we prove that, unless $\mathrm{P}=\mathrm{NP}$, no
polynomial algorithm can approximate the minimum length of \sws\ for
a given \san\ within a constant factor. This result was announced in
the survey~\cite{Vo08} (with a reference to the present author's
unpublished manuscript) but its proof appears here for the first
time. We also mention that a special case of our result, namely,
non-approximability of $\min_{synch}(\mathrsfs{A})$ within factor~2,
was announced by Gawrychowski~\cite{Gawr}.

The paper is organized as follows. First we exhibit an auxiliary
construction that shows non-approximability of
$\min_{synch}(\mathrsfs{A})$ within factor $2-\varepsilon$ for
automata with 3~input letters. Then we show how to iterate this
construction in order to obtain the main result, again for automata
with 3~input letters. Finally, we describe how the construction can
be modified to extend the result also to automata with only 2~input
letters.

\section{Non-approximability within factor~$2-\varepsilon$}

First we fix our notation and introduce some definitions. When we
have specified a DFA $\mathrsfs{A}=\langle Q,\Sigma,\delta\rangle$,
we can simplify the notation by writing $q.w$ instead of
$\delta(q,w)$ for $q\in Q$ and $w \in \Sigma^*$. For each subset
$S\subseteq Q$ and each word $w\in\Sigma^*$, we write $S.w$ instead
of $\{q.w \mid q \in S \}$. We say that a subset $S \subseteq Q$ is
\emph{occupied} after applying some word $v\in\Sigma^*$ if $S
\subseteq Q.v$.

The length of a word $w\in\Sigma^*$ is denoted by $|w|$. If $1\le
s\le|w|$, then $w[s]$ denotes the letter in the $s$-th position of
$w$; similarly, if $1\le s<t\le|w|$, then $w[s..t]$ stands for the
word $w[s]w[s+1]\cdots w[t]$.

Let $\mathcal{K}$ be a class of \sa\ We say that an algorithm $M$
\emph{approximates the minimal length of \sws\ in $\mathcal{K}$} if,
for an arbitrary DFA $\mathrsfs{A}\in\mathcal{K}$, the algorithm
calculates a positive integer $M(\mathrsfs{A})$ such that
$M(\mathrsfs{A})\ge \min_{synch}(\mathrsfs{A})$. The
\emph{performance ratio} of $M$ at $\mathrsfs{A}$ is
$R_M(\mathrsfs{A})=\dfrac{M(\mathrsfs{A})}{\min_{synch}(\mathrsfs{A})}$.
The algorithm is said to \emph{approximate the minimal length of
\sws\ within factor $k \in \mathbb{R}$} if
$$\sup\{R_M(\mathrsfs{A})\mid \mathrsfs{A}\in\mathcal{K}\}=k.$$

Even though the following theorem is subsumed by our main result, we
prove it here because the proof demonstrates underlying ideas in a
nutshell and in the same time presents a construction that serves as
the induction basis for the proof of the main theorem.

\begin{theorem}\label{prop_2}
If $\mathrm{P}\ne\mathrm{NP}$, then for no $\varepsilon>0$ a
polynomial algorithm approximates the minimal length of \sws\ within
factor $2-\varepsilon$ in the class of all \sa\ with $3$ input
letters.
\end{theorem}

\begin{proof}
    Arguing by contradiction, assume that there exist a real number
    $\varepsilon>0$ and a polynomial algorithm $M$ such that
    $R_M(\mathrsfs{A})\le2-\varepsilon$ for every \san\
    $\mathrsfs{A}$ with $3$ input letters.

We fix an arbitrary $n>2$ and take an arbitrary instance $\psi$ of
the classical NP-complete problem SAT (the satisfiability problem
for a system of clauses, that is, formulae in conjunctive normal
form) with $n$ variables. Let $m$ be the number of clauses in
$\psi$. We shall construct a synchronizing automaton
$\mathrsfs{A}(\psi)$ with 3 input letters and polynomial in $m,n$
number of states such that $\min_{synch}(\mathrsfs{A}(\psi))=n+2$ if
$\psi$ is satisfiable and $\min_{synch}(\mathrsfs{A}(\psi))> 2(n-1)$
if $\psi$ is not satisfiable. If $n$ is large enough, namely,
$n\ge\frac6{\varepsilon}-2$, then we can decide whether or not
$\psi$ is satisfiable by running the algorithm $M$ on
$\mathrsfs{A}(\psi)$. Indeed, if $\psi$ is not satisfiable, then
$M(\mathrsfs{A}(\psi))\ge\min_{synch}(\mathrsfs{A}(\psi))> 2(n-1)$,
but, if $\psi$ is satisfiable, then
\begin{multline*}
M(\mathrsfs{A}(\psi))\le(2-\varepsilon)\min\nolimits_{synch}(\mathrsfs{A}(\psi))
=(2-\varepsilon)(n+2)\\ \le(2-\frac6{n+2})(n+2)=2(n-1).
\end{multline*}
Clearly, this yields a polynomial algorithm for SAT: given an
instance of SAT, we can first, if necessary, enlarge the number of
variables to at least $\frac6{\varepsilon}-2$ without influencing
satisfiability and then apply the above procedure. This contradicts
the assumption that $\mathrm{P}\ne\mathrm{NP}$.

Now we describe the construction of the automaton
$\mathrsfs{A}(\psi)=\langle Q,\Sigma,\delta\rangle$. The state set
$Q$ of $\mathrsfs{A}(\psi)$ is the disjoint union of the three
following sets:
\begin{align*}
S^1&= \{q_{i,j} \mid 1 \leq i \leq m+1,\, 1 \leq j \leq n+1,\, i
\neq m+1 \text{ or } j \neq n+1\},\\
S^2&= \{p_{i,j} \mid 1 \leq i \leq m+1,\, 1 \leq j \leq n+1 \},\\
S^3& = \{z_1,z_0\}.
\end{align*}
The size of $Q$ is equal to $2(m+1)(n+1)+1$ so a polynomial in
$m,n$.

The input alphabet $\Sigma$ of $\mathrsfs{A}(\psi)$ is the set
$\{a,b,c\}$. In order to describe the transition function
$\delta:Q\times\Sigma\to Q$, we need an auxiliary function $f
:\{a,b\} \times \{1, \ldots, m\} \times \{1, \ldots, n\} \to Q$
defined as follows. Let the variables involved in $\psi$ be
$x_1,\dots,x_n$ and the clauses of $\psi$ be $c_1,\dots,c_m$. For a
literal $y\in\{x_1,\dots,x_n,\neg x_1,\dots,\neg x_n\}$ and a clause
$c_i$, we write $y \in c_i$ to denote that $y$ appears in $c_i$. Now
set $$f(d,i,j) =
\begin{cases}
    z_0 \text{ if } d=a \text{ and } x_j \in c_i,\\
    z_0 \text{ if } d=b \text{ and } \neg x_j \in c_i,\\
    q_{i,j+1} \text{ otherwise}.
\end{cases}$$

The transition function $\delta$ is defined according to the
following table:

\begin{center}
\begin{tabular}{|l|c|c|c|}
    \hline
    \hfil State $q\in Q$ \hfill    \rule{0pt}{14pt}                      &$\delta(q,a)$&$\delta(q,b)$&$\delta(q,c)$ \\
    \hline
    $q_{i,j} \text{ for } 1\le i\le m, 1\le j\le n \rule{0pt}{14pt}     $&$ f(a,i,j)  $&$ f(b,i,j)  $&$ q_{i,1}  $\\
    \hline
    $q_{m+1,j} \text{ for } 1\le j\le n  \rule{0pt}{14pt}               $&$q_{m+1,j+1}$&$q_{m+1,j+1}$&$ q_{m+1,1}$\\
    \hline
    $q_{i,n+1} \text{ for } 1\le i \le m      \rule{0pt}{14pt}          $&$ z_0       $&$ z_0       $&$ q_{m+1,1}$\\
    \hline
    $p_{i,j}  \text{ for } 1\le i\le m+1,1\le j\le n \ \rule{0pt}{14pt} $&$ p_{i,j+1} $&$ p_{i,j+1} $&$ p_{i,j+1}$\\
    \hline
    $p_{i,n+1} \text{ for } 1\le i \le m+1 \rule{0pt}{14pt}             $&$ z_0       $&$ z_0       $&$ q_{i,1}  $\\
    \hline
    $z_1             \rule{0pt}{14pt}                                   $&$ q_{m+1,1} $&$ q_{m+1,1} $&$ z_0       $\\
    \hline
    $z_0               \rule{0pt}{14pt}                                 $&$ z_0       $&$ z_0       $&$ z_0       $\\
    \hline
\end{tabular}
\end{center}

Let us informally comment on the essence of the above definition.
Its most important feature is that, if the literal $x_j$
(respectively $\neg x_j$) occurs in the clause $c_i$, then the
letter $a$ (respectively $b$) moves the state $q_{i,j}$ to the state
$z_0$. This encodes the situation when one can satisfy the clause
$c_i$ by choosing the value $1$ (respectively $0$) for the variable
$x_j$. Otherwise, the letter $a$ (respectively $b$) increases the
second index of the state. This means that one cannot make $c_i$ be
true by letting $x_j=1$ (respectively $x_j=0$), and the next
variable has to be inspected. Of course, this encoding idea is not
new, see, e.g., \cite{Ep90}.

By the definition, $z_0$ is the zero state of the automaton
$\mathrsfs{A}(\psi)$. Since there is a path to $z_0$ from each state
$q\in Q$,  the automaton $\mathrsfs{A}(\psi)$ is synchronizing.

Figure~\ref{A2_example} shows two automata of the form
$\mathrsfs{A}(\psi)$ build for the SAT instances
\begin{align*}
\psi_1&=\{x_1 \vee x_2 \vee x_3,\, \neg x_1 \vee x_2,\, \neg x_2
\vee x_3,\,\neg x_2 \vee \neg x_3\},\\
\psi_2&=\{x_1 \vee x_2,\,\neg x_1 \vee x_2,\, \neg x_2 \vee
x_3,\,\neg x_2 \vee \neg x_3\}.
\end{align*}
If at some state $q\in Q$ the picture has no outgoing arrow labelled
$d\in\Sigma$, the arrow $q\stackrel{d}{\to}z_0$ is assumed (all
those arrows are omitted in the picture to improve readability). The
two instances differ only in the first clause: in $\psi_1$ it
contains the variable $x_3$ while in $\psi_2$ it does not.
Correspondingly, the automata $\mathrsfs{A}(\psi_1)$ and
$\mathrsfs{A}(\psi_2)$ differ only by the outgoing arrow labelled
$a$ at the state $q_{1,3}$: in $\mathrsfs{A}(\psi_1)$ it leads to
$z_0$ (and therefore, it is not shown) while in
$\mathrsfs{A}(\psi_2)$ it leads to the state $q_{1,4}$ and is shown
by the dashed line.

Observe that $\psi_1$ is satisfiable for the truth assignment
$x_1=x_2=0$, $x_3=1$ while $\psi_2$ is not satisfiable. It is not
hard to check that the word $cbbac$ synchronizes
$\mathrsfs{A}(\psi_1)$ and the word $a^7c$ is one of the shortest
reset words for $\mathrsfs{A}(\psi_2)$.

\begin{figure}[p]
{\normalsize
\begin{center}
\begin{picture}(200,180)(0,-175)
\node(n14)(71.98,-111.8){$q_{3,2}$}
\node(n75)(52.01,-111.8){$q_{2,2}$}
\node(n32)(52.01,-131.74){$q_{2,3}$}
\node(n41)(52.01,-151.68){$q_{2,4}$}
\node(n42)(71.98,-151.68){$q_{3,4}$}
\node(n202)(51.73,-91.68){$q_{2,1}$}
\node(n172)(71.92,-91.86){$q_{3,1}$}
\node(n472)(71.98,-131.74){$q_{3,3}$}
\node(n474)(91.95,-111.8){$q_{4,2}$}
\node(n475)(92.09,-151.68){$q_{4,4}$}
\node(n476)(91.94,-91.91){$q_{4,1}$}
\node(n477)(92.09,-131.74){$q_{4,3}$}
\node(n478)(32.02,-111.8){$q_{1,2}$}
\node(n479)(32.04,-151.68){$q_{1,4}$}
\node(n480)(31.96,-91.86){$q_{1,1}$}
\node(n481)(32.02,-131.74){$q_{1,3}$}
\node(n482)(112.0,-112.0){$q_{5,2}$}
\node(n483)(112.45,-152.54){$z_1$}
\node(n484)(111.87,-91.86){$q_{5,1}$}
\node(n485)(111.93,-131.74){$q_{5,3}$}

\drawedge(n480,n478){$b$}
\drawedge(n478,n481){$b$}
\drawedge[ELdist=1.1](n32,n41){$a,b$}
\drawedge(n472,n42){$b$}
\drawedge[ELdist=1.1](n476,n474){$a,b$}
\drawedge(n474,n477){$a$}
\drawedge(n477,n475){$a$}
\drawedge[ELdist=1.1](n484,n482){$a,b$}
\drawedge[ELdist=1.1](n482,n485){$a,b$}
\drawedge[ELside=r,ELdist=1.1](n485,n483){$a,b$}
\drawedge[ELside=r,ELdist=.1,ELpos=30,curvedepth=-8](n483,n484){$a,b$}
\drawedge(n202,n75){$a$}
\drawedge(n75,n32){$b$}
\drawedge[ELdist=1.1](n172,n14){$a,b$}
\drawedge(n14,n472){$a$}
\drawedge[curvedepth=9.99](n481,n480){$c$}
\drawedge[ELpos=30,curvedepth=3.9,exo=2](n478,n480){$c$}
\drawedge[curvedepth=8.14](n32,n202){$c$}
\drawedge[ELpos=30,curvedepth=3.9,exo=2](n75,n202){$c$}
\drawedge[curvedepth=8.53](n472,n172){$c$}
\drawedge[ELpos=30,curvedepth=3.9,exo=2](n14,n172){$c$}
\drawedge[curvedepth=8.27](n477,n476){$c$}
\drawedge[ELpos=30,curvedepth=3.9,exo=2](n474,n476){$c$}
\drawedge[curvedepth=8.53](n485,n484){$c$}
\drawedge[ELpos=30,curvedepth=3.9,exo=2](n482,n484){$c$}

\drawloop[loopangle=171.25](n480){$c$}
\drawloop[loopangle=169.99](n202){$c$}
\drawloop[loopdiam=7.47,loopangle=165.58](n172){$c$}
\drawloop[loopangle=165.58](n476){$c$}
\drawloop[loopangle=157.89](n484){$c$}

\node(n826)(72.26,-40.08){$p_{3,2}$}
\node(n845)(52.29,-40.08){$p_{2,2}$}
\node(n827)(52.29,-56.05){$p_{2,3}$}
\node(n828)(52.29,-72.02){$p_{2,4}$}
\node(n829)(72.26,-72.02){$p_{3,4}$}
\node(n830)(52.01,-23.93){$p_{2,1}$}
\node(n831)(72.2,-24.11){$p_{3,1}$}
\node(n832)(72.26,-56.05){$p_{3,3}$}
\node(n833)(92.23,-40.08){$p_{4,2}$}
\node(n834)(92.07,-71.94){$p_{4,4}$}
\node(n835)(92.31,-24.11){$p_{4,1}$}
\node(n836)(92.37,-56.05){$p_{4,3}$}
\node(n837)(32.3,-40.08){$p_{1,2}$}
\node(n838)(32.32,-72.02){$p_{1,4}$}
\node(n839)(32.24,-24.11){$p_{1,1}$}
\node(n840)(32.3,-56.05){$p_{1,3}$}
\node(n841)(112.21,-40.08){$p_{5,2}$}
\node(n842)(112.21,-72.02){$p_{5,4}$}
\node(n843)(112.15,-24.11){$p_{5,1}$}
\node(n844)(112.21,-56.05){$p_{5,3}$}

\drawedge(n839,n837){$a,b,c$}
\drawedge(n837,n840){$a,b,c$}
\drawedge[ELdist=3.0](n840,n838){$a,b,c$}
\drawedge[ELdist=1.1](n827,n828){$a,b,c$}
\drawedge(n832,n829){$a,b,c$}
\drawedge[ELdist=1.1](n835,n833){$a,b,c$}
\drawedge(n833,n836){$a,b,c$}
\drawedge(n836,n834){$a,b,c$}
\drawedge[ELdist=1.1](n843,n841){$a,b,c$}
\drawedge[ELdist=1.1](n841,n844){$a,b,c$}
\drawedge[ELdist=1.1](n844,n842){$a,b,c$}
\drawedge(n830,n845){$a,b,c$}
\drawedge(n845,n827){$a,b,c$}
\drawedge[ELdist=1.1](n831,n826){$a,b,c$}
\drawedge(n826,n832){$a,b,c$}
\drawedge(n838,n480){$c$}
\drawedge(n828,n202){$c$}
\drawedge(n829,n172){$c$}
\drawedge(n834,n476){$c$}
\drawedge(n842,n484){$c$}

\node[Nw=10.32,Nh=9.0,Nmr=0.0](n1310)(8.0,-91.91){$x_1$}
\node[Nw=10.32,Nh=9.0,Nmr=0.0](n1316)(8.03,-111.91){$x_2$}
\node[Nw=10.32,Nh=9.0,Nmr=0.0](n1318)(8.03,-131.91){$x_3$}
\node[Nw=10.32,Nh=9.0,Nmr=0.0](n1367)(32.03,-12.0){$c_1$}
\node[Nw=10.32,Nh=9.0,Nmr=0.0](n1368)(52.0,-12.0){$c_2$}
\node[Nw=10.32,Nh=9.0,Nmr=0.0](n1369)(71.97,-12.0){$c_3$}
\node[Nw=10.32,Nh=9.0,Nmr=0.0](n1370)(91.94,-12.0){$c_4$}
\node[Nw=10.32,Nh=9.0,Nmr=0.0](n1525)(111.91,-12.0){$5$}
\node[Nw=10.32,Nh=9.0,Nmr=0.0](n1529)(8.03,-151.91){$4$}

\node(n1646)(112.03,-179.97){$z_0$}

\drawbpedge[ELpos=10,eyo=3.5](n479,-25,135.29,n484,-26,64.96){$c$}
\drawbpedge[ELpos=15,eyo=2](n41,-34,83.0,n484,-26,72.79){$c$}
\drawbpedge[ELpos=15](n42,-41,72.88,n484,-19,41.96){$c$}
\drawbpedge[ELpos=15,eyo=-2](n475,-50,52.03,n484,-6,28.36){$c$}

\drawrect[dash={5.0 3.0}{0.0}](17.72,-17.705,134,-78.295)
\drawrect[dash={5.0 3.0}{0.0}](105.185,-148.555,134,-187.445)

\node[linecolor=White,Nw=10.32,Nh=9.0,Nmr=0.0](n910)(128,-72.0){$S^2$}
\node[linecolor=White,NLdist=0.26,Nw=10.32,Nh=9.0,Nmr=0.0](n912)(136.0,-96.0){$S^1$}
\node[linecolor=White,NLdist=0.26,Nw=10.32,Nh=9.0,Nmr=0.0](n914)(128,-175.97){$S^3$}

\drawedge[dash={3.0
3.0}{0.0},ELside=r,ELdist=2.0,curvedepth=-8.6](n481,n479){$a$ in
$\mathrsfs{A}(\psi_2)$}

\drawedge[curvedepth=8.0](n481,n479){$b$}
\end{picture}
\end{center}
\caption{ The automata $\mathrsfs{A}(\psi_1)$ and
$\mathrsfs{A}(\psi_2)$} \label{A2_example}}
\end{figure}

To complete the proof, it remains to show that
$\min_{synch}(\mathrsfs{A}(\psi))=n+2$ if $\psi$ is satisfiable and
$\min_{synch}(\mathrsfs{A}(\psi))>2(n-1)$ if $\psi$ is not
satisfiable. First consider the case when $\psi$ is satisfiable.
Then there exists a truth assignment
$\tau:\{x_1,\dots,x_n\}\to\{0,1\}$ such that
$c_i(\tau(x_1),\dots,\tau(x_n))=1$ for every clause $c_i$ of $\psi$.
We construct a word $v=v(\tau)$ of length $n$ as follows:
\begin{equation}
\label{encoding} v[j]=\begin{cases}
a &\text{ if } \tau(x_j)=1,\\
b &\text{ if } \tau(x_j)=0.
\end{cases}
\end{equation}
We aim to prove that the word $w=cvc$ is a \sw\ for
$\mathrsfs{A}(\psi)$, that is, $Q.w=\{z_0\}$. Clearly, $z_1.c=z_0$.
Further, $S^2.cv=\{z_0\}$ because every word of length $n+1$ that
does not end with $c$ sends $S^2$ to $z_0$. Now let $T=\{q_{i,1}
\mid 1\le i\le m+1\}$, so $T$ is the ``first row'' of $S^1$. Observe
that $S^1.c=T$. Since $c_i(\tau(x_1),\dots,\tau(x_n))=1$ for every
clause $c_i$, there exists an index $j$ such that either $x_j\in
c_i$ and $\tau(x_j)=1$ or $\neg x_j \in c_i$ and $\tau(x_j)=0$. This
readily implies (see the comment following the definition of the
transition function of $\mathrsfs{A}(\psi)$) that $q_{i,1}.v=z_0$
for all $1\le i\le m$. On the other hand, $q_{m+1,1}.v=z_1$ because
every word of length $n$ that does not involve $c$ sends $q_{m+1,1}$
to $z_1$. Thus, $S^1.cv=T.v=S^3$ and $S^1.w=\{z_0\}$. We have shown
that $w$ synchronizes $\mathrsfs{A}(\psi)$, and it is clear that
$|w|=n+2$ as required.

Now we consider the case when $\psi$ is not satisfiable.
\begin{lemma}
\label{rem_corr}  If $\psi$ is not satisfiable, then, for each word
$v\in\{a,b\}^*$ of length $n$, there exists $i \le m$ such that
$q_{i,n+1} \in T.v$.
\end{lemma}
\begin{proof}
Define a truth assignment $\tau:\{x_1,\dots,x_n\}\to\{0,1\}$ as
follows:
$$\tau(x_j)=\begin{cases}
1& \text{ if } v[j]=a,\\
0& \text{ if } v[j]=b.
\end{cases}$$
Since $\psi$ is not satisfiable, we have
$c_i(\tau(x_1),\dots,\tau(x_n))=0$ for some clause $c_i$, $1\le i\le
m$. According to our definition of the transition function of
$\mathrsfs{A}(\psi)$, this means that $q_{i,j}.v[j]=q_{i,j+1}$ for
all $j=1,\dots,n$. Hence $q_{i,n+1}= q_{i,1}.v\in T.v$.\qed
\end{proof}

\begin{lemma}
\label{rem_corr2}  If $\psi$ is not satisfiable, then for each word
$v\in\{a,b\}^*$ of length~$n$ and each letter $d \in \Sigma$, the
state $q_{m+1,1}$ belongs to $T.vd$.
\end{lemma}

\begin{proof}
If $d=c$, the claim follows from Lemma~\ref{rem_corr} and the
equalities $q_{m+1,1}=q_{i,n+1}.c$ that hold for all $i \leq m$. If
$d\ne c$, we observe that the state $q_{m+1,1}$ is fixed by all
words of length $n+1$ not involving $c$.\qed
\end{proof}

Let $w'$ be a \sw\ of minimal length for $\mathrsfs{A}(\psi)$ and
denote $w=cw'c$. Then the word $w$ is also synchronizing and
$\ell=|w|>n$ because already the length of the shortest path from
$q_{m+1,1}$ to $z_0$ is equal to $n+1$. Let $k$ be the rightmost
position of the letter $c$ in the word $w[1..n]$.

\begin{lemma}
\label{T is occupied} $T\subseteq Q.w[1..k]$.
\end{lemma}

\begin{proof}
Indeed, since $k\le n$, for each $1 \le i \le m+1$ we have
$$p_{i,n+2-k}.w[1..k-1]w[k]=p_{i,n+1}.c = q_{i,1} \in T.\eqno{\Box}$$
\end{proof}

We denote by $v$ the longest prefix of the word $w[k+1..\ell]$ such
that $v\in\{a,b\}^*$ and $|v|\le n$. Since $w$ ends with $c$, the
word $v$ cannot be a suffix of $w$. Let $d\in\Sigma$ be the letter
that follows $v$ in $w$. If $|v|=n$, then Lemma~\ref{rem_corr2}
implies that $q_{m+1,1} \in T.vd$. If $|v|<n$, then by the
definition of $v$ we have $d=c$. Hence
$$q_{m+1,1}.vd=q_{m+1,|v|+1}.c=q_{m+1,1}.$$
Thus, $q_{m+1,1} \in T.vd$ also in this case. Combining this with
Lemma~\ref{T is occupied}, we have
\begin{equation}
\label{occupied} Q.w[1..k]vd\supseteq T.vd \ni q_{m+1,1}.
\end{equation}
From the definitions of $k$ and $v$ it readily follows that
$w[k+1..n]$ is a prefix of $v$ whence $|v|\ge n-k$. Thus,
$|w[1..k]vd|\ge k+(n-k)+1=n+1$. Recall that the length of the
shortest path from $q_{m+1,1}$ to $z_0$ is equal to $n+1$, and the
suffix of $w$ following $w[1..k]vd$ must bring the state $q_{m+1,1}$
to $z_0$ in view of~\eqref{occupied}. Hence $|w|\ge
(n+1)+(n+1)=2n+2>2n$ and $|w'|>2(n-1)$. We have proved that
$\min_{synch}(\mathrsfs{A}(\psi))>2(n-1)$ if $\psi$ is not
satisfiable.\qed
\end{proof}

\section{The main result}

The main result of this paper is
\begin{theorem}\label{theorem_r}
If $\mathrm{P}\ne\mathrm{NP}$, then no polynomial algorithm can
approximate the minimal length of \sws\ within a constant factor in
the class of all \sa\ with $3$ input letters.
\end{theorem}

\begin{proof}
Again we fix an arbitrary $n>2$ and take an arbitrary instance
$\psi$ of SAT with $n$ variables. We shall prove by induction that
for every $r=2,3,\dots$ there exists a \san\
$\mathrsfs{A}_r(\psi)=\langle Q_r,\Sigma,\delta_r\rangle$ with the
following properties:
\begin{itemize}
\item $\Sigma=\{a,b,c\}$;
\item $|Q_r|$ is bounded by a polynomial of $n$ and the number $m$ of
clauses of $\psi$;
\item if $\psi$ is satisfiable under a truth assignment
$\tau:\{x_1,\dots,x_n\}\to\{0,1\}$, then the word
$w=c^{r-1}v(\tau)c$ of length $n+r$ synchronizes
$\mathrsfs{A}_r(\psi)$ (see~\eqref{encoding} for the definition of
the word $v(\tau)$);
\item $\min_{synch}(\mathrsfs{A}_r)>r(n-1)$ if $\psi$ is not
satisfiable.
\end{itemize}
Then, applying the same standard argument as in the proof of
Theorem~\ref{prop_2}, we conclude that for no $\varepsilon>0$ the
minimal length of \sws\ can be approximated by a polynomial
algorithm within factor $r-\varepsilon$. Since $r$ can be
arbitrarily large, the statement of the main result follows.

The induction basis is verified in the proof of
Theorem~\ref{prop_2}: we can choose the \san\ $\mathrsfs{A}(\psi)$
to play the role of $\mathrsfs{A}_2(\psi)$. For the sake of
uniformity, in the sequel we refer to the state set $Q$ of
$\mathrsfs{A}(\psi)$ and its transition function $\delta$ as to
$Q_2$ and respectively $\delta_2$.

Now suppose that $r>2$ and the automaton
$\mathrsfs{A}_{r-1}(\psi)=\langle
Q_{r-1},\Sigma,\delta_{r-1}\rangle$ with the desired properties has
already been constructed. We let
$$Q_r = Q_{r-1}\bigcup(Q_2 \setminus\{z_0\})\times Q_{r-1}.$$
Clearly, $|Q_r|=|Q_{r-1}|\cdot |Q_2|$ and from the induction
assumption it follows that $|Q_r|$ is a polynomial in $m,n$.

We now define the transition function $\delta_r:Q_r\times\Sigma\to
Q_r$. Let $d \in \Sigma$, $q \in Q_r$. If $q \in Q_{r-1}$, then we
set
\begin{equation}
\delta_r(q, d) = \delta_{r-1}(q, d). \label{Def1}
\end{equation}
If $q=(q',q'')\in(Q_2\setminus\{z_0\})\times Q_{r-1}$, we define
\begin{equation}
\label{Def2} \delta_r(q, d) =
\begin{cases}
    z_0     & \text{if } \delta_2(q', d) = z_0,           \\
    q''     & \text{if } \delta_2(q', d) = q_{m+1,1} \text{ and either } \\
            & q'=q_{i,n+1} \text{ for }i\in\{1,\dots,m\} \\
            & \text{or } q'=q_{m+1,j} \text{ for }j\in\{2,\dots,n\}\\
            & \text{or } q'=z_1,  \\
    (\delta_2(q', d), q'') & \text{in all other cases}.
\end{cases}
\end{equation}
Using this definition and the induction assumption, one can easily
verify that the state $z_0$ is the zero state of the automaton
$\mathrsfs{A}_r(\psi)$ and that there is a path to $z_0$ from
every state in $Q_r$. Thus, $\mathrsfs{A}_r(\psi)$ is a \san.

In order to improve readability, we denote the subset
$\{q_{i,j}\}\times Q_{r-1}$ by $Q_{i,j}$ for each state $q_{i,j}
\in S^1$ and the subset $\{p_{i,j}\}\times Q_{r-1}$ by $P_{i,j}$
for each state $p_{i,j} \in S^2$. Slightly abusing notation, we
denote by $T$ the ``first row'' of $S^1 \times Q_{r-1}$, i.e.\
$T=\bigcup_{1\le i\le m+1}Q_{i,1}$. Similarly, let
$P=\bigcup_{1\le i\le m+1}P_{i,1}$ be the ``first row'' of $S^2
\times Q_{r-1}$. We also specify that the dot-notation (like
$q.d$) always refers to the function $\delta_r$.

First we aim to show that if $\psi$ is satisfiable under a truth
assignment $\tau:\{x_1,\dots,x_n\}\to\{0,1\}$, then the word
$w=c^{r-1}v(\tau)c$ synchronizes the automaton
$\mathrsfs{A}_r(\psi)$. By \eqref{Def1} and the induction
assumption we have $Q_{r-1}.c \subseteq Q_{r-1}$ and
$Q_{r-1}.c^{r-2}v(\tau)c=z_0$. Further, we can decompose $((Q_2
\setminus\{z_0\})\times Q_{r-1}).c$ as $\{z_0\}\cup F_{r-1}\cup
F_r$ for some sets $F_{r-1} \subseteq Q_{r-1}$ and $F_r \subseteq
(Q_2 \setminus\{z_0\})\times Q_{r-1}$. By the induction
assumption,
$$F_{r-1}.c^{r-2}v(\tau)c \subseteq Q_{r-1}.c^{r-2}v(\tau)c=z_0$$
Consider the set $F_r$.  Using the definition of the action of $c$
on $Q_2$ via~$\delta_2$, one can observe that $F_r=T\cup G$ where
$G$ stands for $S^2\times Q_{r-1}\setminus P$. From \eqref{Def2}
we see that $T.c=T$ and $G.c \subseteq T\cup G$. Thus we have
$F_r.c^{r-2}v(\tau)c\subseteq T.v(\tau)c\cup G.v(\tau)c,$ and
combining the first alternative in \eqref{Def2} with properties of
the automaton $\mathrsfs{A}_2(\psi)$ established in the proof of
Theorem~\ref{prop_2}, we obtain $T.v(\tau)c =G.v(\tau)c=\{z_0\}$.

Now we consider the case when $\psi$ is not satisfiable. The
following lemma is parallel to Lemma~\ref{rem_corr} and has the same
proof because the action of $a$ and $b$ on the ``blocks'' $Q_{i,j}$
with $1\le i\le m$ and $1\le j\le n$ via $\delta_r$ precisely
imitates the action of $a$ and $b$ on the states $q_{i,j}$ in the
automaton $\mathrsfs{A}(\psi)$, see the last alternative in
\eqref{Def2}.

\begin{lemma}
\label{rem_corr_R}  If $\psi$ is not satisfiable, then, for each
word $v\in\{a,b\}^*$ of length $n$, there exists $i \le m$ such that
$Q_{i,n+1} \subseteq \delta_r(T,v)$.\qed
\end{lemma}

In contrast, the next lemma which is a counterpart of
Lemma~\ref{rem_corr2} uses the fact that in some cases the action of
the letters via $\delta_r$ drops states from $((Q_2 \setminus
{z_0})\times Q_{r-1})$ down to $Q_{r-1}$, see the middle alternative
in \eqref{Def2}.

\begin{lemma}
\label{rem_corr2_R}  If $\psi$ is not satisfiable, then for each
word $v\in\{a,b\}^*$ of length $n$ and each letter $d \in \Sigma$,
we have $Q_{r-1}\subseteq \delta_r(T,vd)$.
\end{lemma}

\begin{proof}
If $d=c$, the claim follows from Lemma~\ref{rem_corr_R} and the
equalities $\delta_r((q_{i,n+1},q''),c)=q''$ that hold for all $i
\le m$ and all $q''\in Q_{r-1}$. If $d\ne c$, we observe that
$\delta_r((q_{m+1,1},q''),v)=(z_1,q'')$ and
$\delta_r((z_1,q''),a)=\delta_r((z_1,q''),b)=q''$ for all $q''\in
Q_{r-1}$. \qed
\end{proof}

Let $w'$ be a \sw\ of minimal length for $\mathrsfs{A}_r(\psi)$ and
denote $w=cw'c$. Then the word $w$ is also synchronizing and
$\ell=|w|>(r-1)n$ by the induction assumption. Let $k$ be the
rightmost position of the letter $c$ in the word $w[1..n]$. We have
the next lemma parallel to Lemma~\ref{T is occupied} and having the
same proof (with the ``blocks'' $P_{i,j}$ with $1 \le i\le m+1$,
$n+2-k\le j\le n$  playing the role of the states $p_{i,j}$).

\begin{lemma}
\label{T is occupied_R} $T\subseteq \delta_r(Q_r,w[1..k])$.\qed
\end{lemma}

Now, as in the proof of Theorem~\ref{prop_2}, we denote by $v$ the
longest prefix of the word $w[k+1..\ell]$ such that $v\in\{a,b\}^*$
and $|v|\le n$. Clearly, $v$ cannot be a suffix of $w$. Let
$d\in\Sigma$ be the letter that follows $v$ in $w$. If $|v|=n$ then
Lemma~\ref{rem_corr2_R} implies that $Q_{r-1}\subseteq
\delta_r(T,vd)$. If $|v|<n$, then by the definition of $v$ we have
$d=c$. Hence
$$\delta_r(Q_{m+1,1},vd)=\delta_r(Q_{m+1,|v|+1},c)=Q_{r-1}.$$
Thus, $Q_{r-1}\subseteq \delta_r(T,vd)$ also in this case. Combining
this with Lemma~\ref{T is occupied_R}, we have
\begin{equation}
\label{occupied_R} \delta_r(Q_r,w[1..k]vd)\supseteq \delta_r(T,vd)
\supseteq Q_{r-1}.
\end{equation}
From the definitions of $k$ and $v$ it readily follows that $|v|\ge
n-k$. Thus, $|w[1..k]vd|\ge k+(n-k)+1=n+1$. The suffix of $w$
following $w[1..k]vd$ must bring the set $Q_{r-1}$ to a single state
in view of~\eqref{occupied_R}. However, by~\eqref{Def1} the
restriction of $\delta_r$ to $Q_{r-1}$ coincides with $\delta_{r-1}$
whence the suffix must be a \sw\ for $\mathrsfs{A}_{r-1}(\psi)$. By
the induction assumption
$\min_{synch}(\mathrsfs{A}_{r-1}(\psi))>(r-1)(n-1)$, and therefore,
$$|w|>(n+1)+(r-1)(n-1)=r(n-1)+2$$ and $|w'|>r(n-1)$. We have thus
proved that $\min_{synch}(\mathrsfs{A}_r(\psi))>r(n-1)$ if $\psi$ is
not satisfiable. This completes the induction step.\qed
\end{proof}

\section{The case of 2-letter alphabets}

We show that the main result extends to \sa\ with only 2 input
letters.

\begin{corollary}
If $\mathrm{P}\ne\mathrm{NP}$, then no polynomial algorithm can
approximate the minimal length of \sws\ within a constant factor
in the class of all \sa\ with $2$ input letters.
\end{corollary}

\begin{proof}
For any \san\ $\mathrsfs{A}=(Q,\{a_1,a_2,a_3\},\delta)$ we can
construct a \san\ $\mathrsfs{B}=(Q',\{a,b\}, \delta')$ such that
\begin{equation}
\label{limits} \min\nolimits_{synch}(\mathrsfs{A}) \le
\min\nolimits_{synch}(\mathrsfs{B}) \le
3\min\nolimits_{synch}(\mathrsfs{A})
\end{equation}
and $|Q'|$ is a polynomial of $|Q|$. Then any polynomial algorithm
approximating the minimal length of \sws\ for 2-letter \sa\ within
factor $r$ would give rise to a polynomial algorithm approximating
the minimal length of \sws\ for 3-letter \sa\ within factor $3r$.
This would contradict Theorem~\ref{theorem_r}.

We let $Q' = Q \ \times \{a_1,a_2,a_3\}$ and define the transition
function $\delta':Q'\times \{a,b\}\to Q'$ as follows:
\begin{align*}
\delta'((q,a_i),a)&=(q,a_{\min(i+1,3)}),\\
\delta'((q,a_i),b)&=(\delta(q,a_i),a_1).
\end{align*}
Thus, the action of $a$ on a state $q'\in Q'$ substitutes an
appropriate letter from in the alphabet $\{a_1,a_2,a_3\}$ of
$\mathrsfs{A}$ for the second component of $q'$ while the action
of $b$ imitates the action of the second component of $q'$ on its
first component and resets the second component to $a_1$. Now is
let a word $w\in\{a_1,a_2,a_3\}$ of length $\ell$ be a \sw\ for
$\mathrsfs{A}$. Define
$$v_s=\begin{cases}
b &\text{if } w[s]=a_1,\\
ab &\text{if } w[s]=a_2,\\
aab &\text{if } w[s]=a_3.
\end{cases}$$
Then the word $v=bv_1\cdots v_\ell$ is easily seen to be a \sw\
for $\mathrsfs{B}$ and $|v|\le 3\ell$ unless all letters in $w$
are $a_3$ but in this case we can just let $a_2$ and $a_3$ swap
their names. Hence the second inequality in~\eqref{limits} holds
true, and the first inequality is clear.
\end{proof}

\subsection*{Acknowledgments}
The author acknowledges support from the Federal Education Agency of
Russia, grant 2.1.1/3537,  and from the Russian Foundation for Basic
Research, grant 09-01-12142.

\end{document}